\documentclass[10pt]{iopart}

\expandafter\let\csname equation*\endcsname\relax
\expandafter\let\csname endequation*\endcsname\relax

\usepackage{amsmath}
\usepackage{amsthm}
\usepackage{amsfonts}
\usepackage{mathtools}
\usepackage{url}
\newcommand{\be}{\begin{equation}}
\newcommand{\ee}{\end{equation}}

\newcommand{\ad}{a^\dagger}

\newcommand{\hH}{\hat{H}}

\newcommand{\Ft}{\mathbb{F}_2}
\newcommand{\Ct}{\mathbb{C}_2}

\newcommand{\mC}{\mathcal{C}}

\newtheorem{thrm}{Theorem}

\begin{document}
\title{Revealing symmetries in quantum computing for many-body systems}

\author{Robert van Leeuwen}

\address{University of Jyv\"askyl\"a, Department of Physics, Nanoscience Center, Jyv\"askyl\"a, Finland}
%\ead{customerservices@ioppublishing.org}
\vspace{10pt}
%\begin{indented}
%\item[]August 2017 (minor update March 2024)
%\end{indented}

%\maketitle
%\tableofcontents
%\section{}
%\subsection{}

\begin{abstract}
We develop a method to deduce the symmetry properties of many-body Hamiltonians when they are prepared in Jordan-Wigner form 
for evaluation on quantum computers.  Symmetries, such as point-group symmetries in molecules, are apparent in the standard second quantized form of the Hamiltonian.
They are, however,
masked when the Hamiltonian is translated into a Pauli matrix representation required for its operation on qubits.
To reveal these symmetries we prove a general theorem that provides a straightforward method to calculate the transformation of Pauli tensor strings under symmetry operations.
They are a subgroup of the Clifford group transformations and induce a corresponding group representation inside the symplectic matrices.
We finally give a simplified derivation of an affine qubit encoding scheme which allows for the removal of qubits due to Boolean symmetries and thus reduces computational effort in quantum computing applications.
\end{abstract}

\section{Introduction}

Recent years have seen steady progress in the design of quantum computing devices with an increasing number of qubits \cite{Moll2018}. These developments are especially interesting for applications
in many-body physics \cite{Fauseweh2024} and quantum chemistry \cite{Bauer2020} as quantum computers are intrinsically advantageous to solve quantum many-body problems. Moreover instead of the exponential computational
cost required to perform such calculations on classical computers, on quantum computers the cost is envisioned to grow only polynomially. This has spurred many efforts in the quantum chemistry \cite{Bauer2020} and
condensed matter physics \cite{Fauseweh2024} community to develop schemes to solve interesting many-body problems on a quantum computer. \\
For the simulation of such quantum systems on a quantum computer we need to represent the many-body Hamiltonian as a linear operator on a system of qubits. This is usually done by first writing the 
Hamiltonian in a second quantized form, i.e. as an operator in Fock space, and then use the Jordan-Wigner transformation \cite{JordanWigner1928,Whitfield2011,Steudtner2018} to translate the creation and annihilation operators into appropriate tensor products of Pauli matrices that can act
on systems of qubits.  The anti-symmetrized kets or Slater determinants in Fock space are then mapped to tensor products of qubits, while the anti-symmetry properties of the Fock space states are now encoded into the Pauli product form of the Hamiltonian. Although this representation is exceedingly useful for practical calculations on quantum computers the Jordan-Wigner transformed representation of the Hamiltonian generally conceals the symmetry properties
of the original second quantised form.  One of the main results of this work is the derivation of a practical way to recover these symmetries directly from the Jordan-Wigner form of the Hamiltonian. \\
The preservation of symmetries in quantum computing is often essential in attaining the correct computational output. As was pointed out in \cite{Ryabinkin2018} there is often no convenient basis in the space of
multiple qubits that is also an eigenbasis of a relevant symmetry operator and when the symmetry is not imposed there is a risk of mixing in the wrong symmetry states leading to a meaningless result. For this and other reasons the related question of enforcing the right symmetries in a calculation has been studied in several works \cite{Ryabinkin2018,Bravyi2017,Yen2019,Setia2020,Cao2022,Picozzi2023}.
 Apart from being of intrinsic scientific interest, the elucidation of symmetries also has important practical applications in the reduction of the number of qubits needed for quantum computations and therefore several techniques have been developed to this aim. A recent scheme was developed by Picozzi and Tennyson \cite{Picozzi2023} which uses a clever affine encoding technique to address the issue.
In this work we give a considerably shorter proof of their main result and illustrate all our results with the elucidating example of the Hubbard dimer, a system that was studied before in the context of qubit reduction using projection 
and shift operators \cite{Moll2016}. We finally like to point out that recently an iterative scheme based on conserved charges has been implemented for qubit reduction in quantum computations \cite{Gunderman2024}. However, unlike the work in \cite{Picozzi2023}, it does not allow for an easy identification of the symmetry labels that characterize the many-body systems at hand.  \\
The paper is structured as follows:
we first give a discussion of the permutational symmetries of Hamiltonians defined in a local atomic basis, such as Hubbard lattices, and review their symmetry properties. Next we illustrate, by means of an example,
the masking of these symmetries after Jordan-Wigner transformation and subsequently prove a general theorem to resolve this problem.  We then address the issue of qubit reduction by symmetries as was discussed in  \cite{Picozzi2023}
in which we make a small correction to their recent proof and provide a much shorter proof of the main result. We finally present an application and our conclusions.

\section{Revealing symmetries in Jordan-Wigner transformed Hamiltonians}

\subsection{Permutation symmetries in many-body Hamiltonians}

In quantum chemical and condensed matter applications one commonly studies second quantized molecular Hamiltonians of the general form
\be
\hH = \sum_{ij}^M h_{ij} \ad_i a_j + \sum_{ijkl}^M v_{ijkl} \, \ad_i \ad_j a_k a_l
\label{H_gen}
\ee
where $a_i^\dagger$ and $a_i$ are creation and annihilation operators for electrons in one-particle state $i$ \cite{Helgaker2000,Stefanucci2013}.
%where $h_{ij}$ and $v_{ijkl}$ are one- and  two-electron integrals over suitable chosen set of orthonormal orbitals.
%\begin{align}
%h_{ij}  &= \int dr \varphi_i^* (r) h(r, \underline{R}) \varphi_j (r) \nonumber \\
%v_{ijkl} &= \int dr dr^\prime \varphi_i^* (r) \varphi_j^* (r^\prime) v(r,r^\prime)\varphi_k^* (r^\prime) \varphi_l^* (r) 
%\end{align}
%where $v(r, r^\prime)=1/|r-r^\prime|$ is the Coulomb interaction and  $\uR=( R_1, \ldots, R_n )$ is the set of nuclear positions.
Let us to be concrete assume that the Hamiltonian represents a molecule in a localized basis, such as L\"owdin orthogonalized
atomic orbitals \cite{Lowdin1950,Helgaker2000}, and that the coefficients $h_{ij}$ and $v_{ijkl}$ represent one- and two-electron integrals over such orbitals.
If the molecule has a non-trivial point group, then every
point group operation will permute a given local orbital from one atomic site to another. 
Since the Hamiltonian of the molecule is invariant under point group symmetries it then follows that for such a permutation we have the following identities:
\be
h_{P(ij)} = h_{ij} \quad \quad v_{P(ijkl)} = v_{ijkl}
\label{P_sym}
\ee
where $P$ is a permutation of the labels of the localized orbitals and
and we used the notation $P(i_1, \ldots, i_M) =P(i_1) \ldots P (i_M)$. 
Equivalently we can write that the Hamiltonian of Eq.(\ref{H_gen}) has the symmetry
\be
\hH = \sum_{ij}^M h_{ij} \ad_{P(i)} a_{P(j)} + \sum_{ijkl}^M v_{ijkl} \, \ad_{P (i)} \ad_{P(j)} a_{P(k)} a_{P(l)}
\label{H_perm}
\ee
for any permutation $P$ of the localized orbitals compatible with the point group symmetry. In this work we study the general class of Hamiltonians that have symmetries of
the kind displayed in Eq.(\ref{H_perm}), whether they are derived from a molecular model or not, as may be the case for certain model Hamiltonians in condensed matter physics.\\
By making various simplifying assumptions about the one-electron matrix elements $h_{ij}$ and the
two-electron integrals $v_{ijkl}$ in Eq.(\ref{H_gen}) one can derive various classes of popular model Hamiltonians that are widely studied in chemistry and 
physics. In chemistry this encompasses, for example, the Parr-Pariser-Pople Hamiltonians \cite{Linderberg1968,Jug1990,Stefanucci2013} which are used in the study of complex unsaturated molecules
while in condensed matter physics an important class consists of the Hubbard Hamiltonians \cite{HubbardReview2022} which are widely studied as they exhibit a metal-insulator transition and
may offer insights in the currently still marginally understood phenomenon of high-temperature superconductivity. 
The general form of the Hubbard Hamiltonian is 
\be
\hH = -t \sum_{\langle i,j \rangle,\sigma} (a_{i\sigma}^\dagger  a_{j\sigma} + a_{j\sigma}^\dagger  a_{i\sigma} )+ U \sum_i \ad_{i \uparrow} a_{i \uparrow}
\ad_{i \downarrow} a_{i \downarrow}
\ee
where $<i,j >$ is a notation for nearest neighbour pairs in a lattice of a certain symmetry, and $\sigma \in\{ \uparrow, \downarrow\}$ denotes the electron spin. The Hubbard systems are studied as standard models for strongly correlated electrons  \cite{HubbardReview2022} and consequently it
has been a long standing goal to solve these models for sufficiently large number of lattice sites to exhibit the behaviour of strongly correlated extended systems.
It is therefore not surprising that in the quantum computing community there has been a large interest in these systems \cite{Reiner2019,Cade2020,Suchsland2022,Stanisic2022} since
a quantum computer may be a novel way to reach that long standing goal and indeed there have been important recent advances \cite{Suchsland2022,Stanisic2022} in that direction.\\
To solve a Hubbard cluster, or generally a Hamiltonian of the form of Eq.(\ref{H_gen}), on a quantum computer one needs to translate the terms in the Hamiltonian from creation and annihilation operators to tensor products of Pauli matrices as is routinely done
with the Jordan-Wigner transformation \cite{JordanWigner1928,Whitfield2011,Steudtner2018}. We can then ask what happens to the permutation symmetries of Eq.(\ref{H_perm})
for the new representation of the Hamiltonian; in other words how do the various Pauli operators transform among one another under the permutation symmetries?
Answering this question is one of the main topics of this work, in which we will derive a general law for the transformation of Pauli operators under symmetries. However, before addressing the general case it is helpful to
clarify the issue using an example, which we will present in the subsequent section.

\subsection{An illustrative example}
\label{sec_example}

To elucidate the general problem we start by giving an example that is as simple as possible but still sufficiently interesting to illustrate the key 
aspects and to which we can refer back later when discussing the general case.
The example is that of spin-less fermions on a triangular cluster \cite{Penz2021} of three atoms described by the Hamiltonian
\be
\hH = a_1^\dagger a_2 + a_2^\dagger a_1 +  a_1^\dagger a_3 + a_3^\dagger a_1  + a_2^\dagger a_3 + a_3^\dagger a_2 
\ee 
where $a_i^\dagger$ and $a_i$ are the creation and annihilation operators for a particle on site $i$.
For this system any permutation $(1,2,3) \rightarrow (P(1), P(2)), P(3))$ of the labels of the creation and annihilation operators 
leaves the Hamiltonian invariant, which constitutes the symmetry group $S_3$, or $C_{3v}$ in point group language.
Let us study this Hamiltonian in the $2$-particle subspace of Fock space \cite{Stefanucci2013}
spanned by the basis vectors
\be
 | 12 \rangle = a_2^\dagger a_1^\dagger |0 \rangle \quad \quad | 13 \rangle = a_3^\dagger a_1^\dagger |0\rangle \quad \quad  | 23 \rangle = a_3^\dagger a_2^\dagger |0 \rangle
 \label{two_basis}
\ee 
and consider the cyclic permutation $P(123)=(231)$. This permutation of labels induces a permutation $U(P) |ij \rangle = |P(ij) \rangle$ on the 2-particle states given by 
\be
U(P) | 12 \rangle= | 2 3 \rangle \quad \quad U(P) | 13 \rangle = |2 1\rangle = - | 12 \rangle \quad \quad U(P) | 23 \rangle = |31 \rangle = - | 13 \rangle
\label{basis_perm}
\ee
where the minus signs in the equations arise from the anti-commutation properties of the creation operators in Eq.(\ref{two_basis}).
The matrix representations $H$ of the Hamiltonian $\hH$ and of $U(P)$ of Eq.(\ref{basis_perm}) in the basis $\{|12 \rangle, |13 \rangle, |23 \rangle\}$ are 
readily calculated to be
\be
H =
\left(
\begin{array}{ccc}
0 & 1 & -1 \\
1 & 0 & 1 \\
-1 & 1 & 0
\end{array}
\right)
\quad \quad 
U (P) 
=
\left(
\begin{array}{ccc}
0 & -1 & 0 \\
0 & 0 & -1 \\
1 & 0 & 0
\end{array}
\right)
\label{UP}
\ee
where $U(P)$ is a unitary matrix which we can check to be commuting with $H$, i.e.
\be
[H , U (P) ] = H U (P)   - U (P) H  = 0 \quad  \Rightarrow \quad  H = U (P)  H U^\dagger (P) 
\ee
and therefore $H$ is invariant under a unitary conjugation induced by $P$, which precisely is what defines a symmetry operation for $H$.
Similar considerations can be done in the one-particle subspace of dimension 3 and the zero-and three particle subspaces of dimension 1.
This gives commuting blocked matrices $H= H_0 \oplus H_1 \oplus H_2 \oplus H_3$ and $U(P)= U(P)_0 \oplus U(P)_1 \oplus U(P)_2 \oplus U(P)_3$ (where the subindex 
of the blocks refers to the particle number) acting in the full $1+3+3+1=8$-dimensional Fock space. In general symmetries are represented by signed permutation matrices $U(P)_n$ that
commute with the Hamiltonian in each sector of fixed particle number $n$.
\\
Let us now translate the problem to qubit space where we directly consider the 8-dimensional space spanned by the 3-qubit tensor products $| \nu \rangle=| \nu_1 , \nu_2 , \nu_3 \rangle
=| \nu_1 \rangle \otimes | \nu_2 \rangle \otimes | \nu_3 \rangle$ with $\nu_i \in \{ 0,1\}$. Let now $X_i,Y_i,Z_i$ be the standard Pauli matrices acting on qubit $i$ and further denote
$\sigma_i^\pm =(X_i \pm i Y_i)/2$. Then the creation operators in Jordan-Wigner form  \cite{JordanWigner1928,Whitfield2011,Steudtner2018} are given by 
\begin{eqnarray}
A_1^\dagger &=&  \sigma_1^- \otimes Z_2 \otimes Z_3 \nonumber \\
A_2^\dagger &=&  1_1 \otimes \sigma_2^- \otimes Z_3 \nonumber \\
A_3^\dagger &=&  1_1 \otimes 1_2 \otimes \sigma^-_3 \nonumber 
\end{eqnarray}
and the corresponding annihilation operators follow from the replacement $\sigma^- \rightarrow \sigma^+$ in these expressions. From this we can then
construct the corresponding form of the Hamiltonian and obtain the expression
\be
\hH= \frac{1}{2} [  X_1 X_2 ( 1-Z_1 Z_2 ) + X_1 Z_2 X_3 (1- Z_1 Z_3)  + X_2 X_3 (1-Z_2 Z_3) ]
\label{ham_JW}
\ee
where we used $Y_j = i X_j Z_j$ to write the Hamiltonian solely in terms of Pauli $X$ and $Z$ matrices.
This Hamiltonian does not have an obvious symmetry under relabelling of the indices. To find the symmetries we first need to construct the equivalent 
of the transformation $U (P)$ but acting on tensor products $| \nu \rangle$. We denote the corresponding operator by $\mC_P$ and discuss its general construction in detail
in the next section. If we use $\mC_P$ (vide infra) we find that the $X_i$ and $Z_i$ operators transform according to
\be
\begin{split}
\mC_P X_1 \mC_P^\dagger &=  X_2 Z_1    \\
\mC_P X_2 \mC_P^\dagger &=  X_3 Z_1  \\
\mC_P X_3 \mC_P^\dagger &= X_1 Z_2 Z_3    
\end{split}
\quad \quad \quad \quad
\begin{split}
\mC_P Z_1 \mC_P^\dagger &=  Z_2   \\
\mC_P Z_2 \mC_P^\dagger &= Z_3   \\
\mC_P Z_3 \mC_P^\dagger &= Z_1  
\end{split} 
\label{XZ_trans} 
\ee
We thus see that the $Z_i$ operators transform cyclically according to the permutation $P$ but that the $X_i$ operators instead transform
to a product of $X$ and $Z$ operators.
We can then check that under these transformations we indeed find an invariance of the Hamiltonian, i.e.
\be
\hH = \mC_P \hH \mC_P^\dagger
\ee
and therefore the transformations in Eq.(\ref{XZ_trans}) reveal a symmetry that was hidden in the explicit form of the Hamiltonian of Eq.(\ref{ham_JW}).
The actual structure behind the transformations in Eq.(\ref{XZ_trans}) remains mysterious at this point, but some insight is attained by writing them in the form of
a so-called Clifford matrix or tableau \cite{Aaronson2004,VanDenBerg2021,Mastel2023} which we will define now. If 
\be
\mC_P X_1^{r_1} \ldots X_M^{r_M} Z_1^{s_1} \ldots Z_M^{s_M} \mC_P^\dagger =   X_1^{r_1^\prime} \ldots X_M^{r_M^\prime} Z_1^{s_1^\prime} \ldots Z_M^{s_M^\prime} 
\label{mC_trans}
\ee
where $r,s$ and $r^\prime,s^\prime$ are $M$-dimensional vectors with entries $0$ and $1$ then the matrix $\mathcal{M}_P$ that maps vector $(r,s)$ to vector $(r^\prime,s^\prime)$ is called
the Clifford matrix or tableau \cite{Aaronson2004,VanDenBerg2021,Mastel2023}; in general it is defined in a similar way for more general unitary transformations than $\mC_P$ that we use here, but the definition above suffices at this point. In our example $\mathcal{M}_P$ has the explicit form
\be
\mathcal{M}_P =
\left(
\begin{array}{ccc|ccc}
0 & 0 & 1 & & & \\
1 & 0 & 0 & & & \\
0 & 1 & 0 & & & \\
\hline
1 & 1 & 0 & 0 & 0 & 1 \\
0 & 0 & 1 & 1 & 0 & 0 \\
0 & 0 & 1 & 0 & 1 & 0
\end{array}
\right)
= 
\left(
\begin{array}{c|c}
\Pi_P & 0 \\
\hline
Q_P & \Pi_P
\end{array}
\right)
\label{tableau_example}
\ee
where the empty block is just filled with zero entries.
The upper left and lower right blocks are readily identified as the permutation matrix $\Pi_P$ for our symmetry,
i.e. $\Pi_{P,ij}$ has an entry equal to $1$ when $i=P(j)$ and has zero entries otherwise.
It remains to explain the structure of the lower left block $Q_P$ of the matrix $\mathcal{M}_P$; we will show that it is 
given by
\be
 Q_P =  L \Pi_P + \Pi_P L
 \label{QP_eqn}
 \ee
 where $L$ is the triangular matrix
 \be
 L =
 \left(
 \begin{array}{ccc}
 0 & 0 & 0 \\
 1 & 0 & 0 \\
 1 & 1 & 0 
 \end{array}
 \right)
 \ee
 in which all entries below the diagonal are filled with ones and the remaining entries with zeroes. 
Eq.(\ref{QP_eqn}) is a special case of a general theorem valid for arbitrary $M$-qubit Hamiltonians. The derivation of that result is the content of the next sections.

\subsection{Mapping symmetries from Fock space to multi-qubit space}

The mapping from Fock space states to occupation number or multi qubit states was already studied by Jordan and Wigner \cite{JordanWigner1928} and has been
found to be great practical use in quantum computing applications. Here, however, we focus on an aspect of this mapping that has not received much attention.
Since symmetries correspond to permutations of the single particle states, they lead to permutations of anti-symmetrized kets of Slater determinants
in Fock space where the permutations typically introduce sign factors.
On the other hand, multi-qubit space consists of tensor products of qubits that have no particular symmetry
under permutations. This means that in the translation from a Fock space to a multi-qubit representation, 
those sign factors have to be introduced explicitly and in the following we will describe how to do this.
As it is easy to loose track of factors and meanings of the quantum states we have tried to make the notation and formulation as clear and precise as possible.
The final result of this section is a precise general definition of the operator $\mC_P$ that we used in Eq.(\ref{XZ_trans}) for our motivating example.
\\
Let $J=(j_1, \ldots, j_N)$ be a multi-index and consider the many-particle states \cite{Stefanucci2013}
\be
| J \rangle  = | j_1, \ldots, j_N \rangle
\ee
representing a fermionic state with one particle in state $j_1$, another in $j_2$, etc. Such states are anti-symmetric, i.e. for any permutation $P \in S_N$ of $N$ labels
we have
\be
| P(J) \rangle = | j_{P(1)}, \ldots, j_{P (N)} \rangle = (-1)^{|P|} | j_1, \ldots, j_N \rangle  
\label{perm_N}
\ee
where $|P|$ is the parity of the permutation, and therefore to choose 
a linearly independent set of many-body kets it is useful to define kets $| J \rangle$ with an ordering, like $j_1 < j_2 < \ldots < j_N$.
There are $M \choose N$ such ordered kets which form an orthonormal basis.
We then define the creation operators $\ad_i$ to create such many-body kets. If we start with the empty state $| 0 \rangle$ we have \cite{Stefanucci2013}
\begin{eqnarray}
|i_1 \rangle &=& \ad_{i_1} | 0 \rangle  \nonumber \\
| i_1, i_2 \rangle &=& \ad_{i_2}  |i_1 \rangle  \nonumber \\
|i_1, \ldots , i_N \rangle &=& \ad_{i_N} | i_1, \ldots, i_{N-1} \rangle = \ad_{i_N} \ldots \ad_{i_1} | 0\rangle
\label{create}
\end{eqnarray}
The action of the corresponding annihilation operators $a_i$ follows from the definition of the adjoint of an operator and is found to be \cite{Stefanucci2013}
\be
 a_{i} | j_1, \ldots, j_N \rangle=  \sum_{l=1}^N (-1)^{N+l}  \delta_{i, j_l} | j_1 \ldots j_{l-1}, j_{l+1}, \ldots j_N \rangle
 \label{ai_action}
\ee
i.e. removing $j_N$ has a plus sign, removing $j_{N-1}$ a minus etc. and this continues in an alternating way. The
Hilbert space 
set of anti-symmetric $N$-particle kets on $M$ is often denoted
$\mathcal{H}_N= \Lambda^N \mathcal{H}_1$ \cite{Penz2021}, i.e. the $N$-fold wedge product of one-dimensional tensors. The full
Fock space is
\be
\mathcal{F} = \mathcal{H}_0 \oplus  \mathcal{H}_1 \oplus \ldots \oplus  \mathcal{H}_M \quad  \quad \quad \dim \mathcal{F} = \sum_{k=0}^M { M \choose k} = 2^M
\ee
The space relevant for quantum computing consists of all linear combinations of the $M$-fold tensor products of qubits, defined more precisely as
\be
\Ct^M = \{ \textrm{span} ( |\nu_1 \rangle \otimes \ldots \otimes | \nu_M \rangle ) | \nu_j \in \{ 0,1 \} \} \quad \quad \dim \Ct^M = 2^M
\ee
As opposed to the anti-symmetric states $| J \rangle$ in Fock space the states in $\Ct^M$ have no particular permutation symmetry and therefore this feature
needs to be build in the operators acting on them.
The basis states in $\mathcal{F}$ and in $\Ct^M$ can related in a one-to-one manner by using an occupation number representation of the one-particle states representing 
the anti-symmetric kets in $\mathcal{F}$.
With $| J \rangle, j_1 < \ldots < j_N$
we associate a vector $\nu= ( \nu_1, \ldots , \nu_{M} )$ for which $\nu_{j_1}=\nu_{j_2}= \ldots =\nu_{j_N}=1$ while all other elements of $\nu$ are zero. 
We then define the one-to-one mapping $\mathcal{J}: \mathcal{F} \rightarrow \Ct^M$ on basis states by
\be
\mathcal{J} (| j_1, \ldots,  j_N \rangle  ) = | \nu_1 \rangle \otimes | \nu_2 \rangle \otimes \ldots \otimes | \nu_M \rangle = | \nu_1 \ldots \nu_M \rangle
\ee
and use linear extension to define the mapping on all of $\mathcal{F}$, such that superpositions of anti-symmetric kets are mapped to superpositions of
multi-qubit states.
For example if $M=3, N=2$ then
\be
\mathcal{J} ( |12 \rangle  + |13 \rangle) = |1 \rangle \otimes | 1 \rangle \otimes | 0 \rangle 
+  |1 \rangle \otimes | 0 \rangle \otimes | 1 \rangle = |110 \rangle + |101 \rangle
\ee
The translation of the creation operators $a_k^\dagger : \mathcal{F} \rightarrow \mathcal{F}$ to operators $A_k^\dagger : \Ct^M \rightarrow \Ct^M$,
and similarly for the annihilation operators,
is furnished by the Jordan-Wigner transformation in which
the creation and annihilation operators are represented as $2^M \times 2^M$ matrices acting on $\Ct^M$ vectors \cite{JordanWigner1928,Whitfield2011,Steudtner2018}:
\begin{eqnarray}
A_k &= 1_1 \otimes \ldots \otimes 1_{k-1} \otimes \sigma_k^+ \otimes Z_{k+1} \otimes \ldots \otimes Z_M 
\label{JW_a}\\
A_k^\dagger &= 1_1 \otimes \ldots \otimes 1_{k-1} \otimes \sigma_k^- \otimes Z_{k+1} \otimes \ldots \otimes Z_M 
\label{JW_ad}
\end{eqnarray}
which allows for the translation of any second quantized Hamiltonian as in (\ref{H_gen}) to a qubit Hamiltonian of the form
\be
\hH = \sum_{ij}^M h_{ij} A_i^\dagger A_j + \sum_{ijkl}^M v_{ijkl} \, A_i^\dagger A_j^\dagger A_k A_l
\label{H_qubit}
\ee
which, when written out in terms of Pauli matrix tensor products, represents
a $2^M \times 2^M$-matrix acting on multi-qubit states
(note that our definition corresponds to the ordering of the labels as in Eq.(\ref{create}), reversing the order gives another prescription).\\
Our discussion so far aimed to set the stage for the question how symmetry operations are mapped from Fock space to multi-qubit space.
In  Eq.(\ref{perm_N}) we considered permutations that only reordered the contents of the ket. Now we extend this to a permutation
$P \in S_M$ of all $M > N$ labels and again we denote
\be
| P (J) \rangle = | j_{P(1)} \ldots j_{P(N)} \rangle 
\ee
which is of the same form except that now there are in general labels in $P(J)$ that do not occur in $J$.
For example for $M=3$, $N=2$ if $|J \rangle = |13 \rangle$ and $P(123)=(231)$ then $| P(13) \rangle = |21\rangle =- |12\rangle$.
As this example shows, the labels $| P(J) \rangle$ may not be in ascending order, but we can reorder them using a permutation of $N$ elements. The actual form
of this permutation is not relevant, but its sign is important. Let the labels $P (J)$ be reordered in ascending order to $\bar{J}$, i.e. $\bar{j}_k \in P(J)$ with
$\bar{j}_1 < \ldots < \bar{j}_N$. Then
 \be
 | P(J ) \rangle = (-1)^{\pi_P (J)} | \bar{J} \rangle
 \label{P_perm}
 \ee
 where $\pi_P (J)$ is the parity (even or odd) of the permutation of $N$ elements that reorders $P(J)$ to $\bar{J}$. 
 Its value can be calculated as follows: we first construct the $M \times M$ permutation matrix $\Pi_{P,ij}$ with entries
 equal to one if $i=P(j)$ and zero entries otherwise. Then we consider the $N \times N$-submatrix of $\Pi_P$, which we denote by
 $(\bar{J}| J)$,  having rows $\bar{J}$ and columns $J$. Then $\pi_P (J)$ is equal
 to the number or row swaps needed to transform $(\bar{J} | J)$ to the $N \times N$ identity matrix. Equivalently we have
 \be
 (-1)^{\pi_P (j)} = \det (\bar{J}| J)
 \ee
i.e. the determinant of the submatrix $(\bar{J} | J)$.
For example, for the cyclic permutation $P(123)=(231)$ we have
\be
\Pi_P=
\left(
\begin{array}{ccc}
0 & 0 & 1 \\
1 & 0 & 0 \\
0 & 1 & 0 
\end{array}
\right)
\quad \quad \quad 
(\overline{13}|13) = (12|13) =
\begin{pmatrix}
0 & 1 \\
1 & 0
\end{pmatrix} 
\ee
Clearly we need only one row swap to convert $(12|13)$ to the identity matrix and its determinant is equal to $-1$
which agrees with the sign in our example where $| P(13) \rangle =- |12\rangle$.\\
With this preparation we are ready to generalize the derivation of our example in the previous section.
Our starting point is the relation for  $N$-electron multi-indices $K$ and $L$
\be
\langle K | \hH | L \rangle = \langle P(K) | \hH | P (L) \rangle 
\ee
where $P$ is a permutation of the $M$ labels of the creation and annihilation operators in the second quantized form of $\hH$ that keeps $\hH$ invariant, as in Eq.(\ref{H_perm}).
The relation above is a simple consequence of the fact that, written out in terms of creation and annihilation operators as in Eq.(\ref{create}), the left and
the right hand side only differ in a renaming of all operators which can not affect the value of the expectation value.
We thus have
\begin{eqnarray}
\langle K | \hH | L \rangle &=& \langle P(K) | \hH | P (L) \rangle 
 = \sum_{I,J} \langle P(K) | I \rangle \langle I | \hH | J \rangle \langle J | P(L) \rangle \nonumber \\
 &=& \sum_{I,J} U_{KI} (P) \langle I | \hH | J \rangle  U_{JL}^\dagger (P)
 \label{H_invariant}
\end{eqnarray}
where we sum over ordered multi-indices $I,J$ and we defined the unitary matrix
\be
U_{KI} (P)= \langle P(K) | I \rangle = (-1)^{\pi_P (K)} \langle \bar{K} | I \rangle =  (-1)^{\pi_P (K)} \delta_{ \bar{K}  I } 
\label{u_def}
\ee 
where $\delta_{ \bar{K}  I }$ is a multi-index Kronecker delta.
The mapping $P \rightarrow U (P)$ is an (in general reducible) unitary representation of the symmetry group of the Hamiltonian
and we have from Eq.(\ref{H_invariant}) that
\be
\hH = U (P) \hH U^\dagger (P) \quad \quad \Rightarrow \quad \quad [ \hH , U (P)] = 0
\label{H_inv}
\ee
for all symmetries described by permutations $P$. If we evaluate Eq.(\ref{u_def}) for the example in the previous section we exactly
recover the matrix $U(P)$ of Eq.(\ref{UP}).\\
We now again consider the effect of a symmetry permutation in multi qubit basis. Our goal is to derive an operator $\mC_P$
that leads to similar equation to (\ref{H_inv}) for the Hamiltonian in Jordan-Wigner form acting on multi-qubit states, i.e. we search for
an operator $\mC_P$ defined on $\Ct^M$ such that
\be
\hH = \mC_P \hH \mC_P^\dagger \quad \quad \Rightarrow \quad \quad [ \hH , \mC_P ] = 0
\ee
for every permutation $P$ corresponding to a symmetry of the Hamiltonian.

The starting point of our definition of this operator is Eq.(\ref{P_perm}).
 If a permutation transforms anti-symmetric kets $|J \rangle$ to
 $| P (J) \rangle$, then in ket $| J \rangle$ the states $j_1, \ldots, j_N$ are occupied, i.e. $\nu_{j_1}= \ldots = \nu_{j_N}=1$ and all other occupation numbers zero, while 
 in $| \bar{J} \rangle$ all states $P(j_1), \ldots, P(j_N)$ are occupied, i.e. $\nu_{P(j_1)}= \ldots = \nu_{P(j_N)}=1$ and all other occupation numbers zero.
 The mapping $| J \rangle \rightarrow | P (J) \rangle$ can therefore be represented in occupation number language as a mapping
 \be
 \mC_P |\nu \rangle = (-1)^{\tau_P (\nu) } | \Pi_P \nu \rangle
 \ee
where $\Pi_P$ is a permutation matrix with $\Ft$ entries defined by $\Pi_{P,ij}=1$ if $i=P(j)$ and zero otherwise
(this is a slight abuse of notation as we defined $\Pi_P$ before but with integer entries rather than integers mod 2, but this does generally not lead to confusion as in practice
the matrix forms are identical). 
It remains to specify the value of the parity $\tau_P (\nu)$ for a given $\nu$ and $P$. This must correspond to the number of row swaps in $(P(J)|J)$ given by $\pi_P (J)$
in the equivalent Fock space expression (\ref{P_perm}), which 
we now must translate to an occupation number expression that we can calculate given $\nu$ and $P$.
If positions $J$ in ascending order in vector $\nu$ have occupation $1$, then we define $[\nu]=J$. For example, if $| \nu \rangle =|101\rangle$ then $[101]=13$.
We then define $\tau_P (\nu)$ to be the number of row swaps in submatrix $([\Pi_P \nu] | [\nu])$ of $\Pi_P$ required to transform it to the identity
matrix, or equivalently
\be
(-1)^{\tau_P (\nu)} = \det ([\Pi_P \nu] | [\nu])
\label{tau_def}
\ee
 where $[\nu]=J$ is a set of labels for the occupied states in $\nu$ and $[\Pi_P \nu]=P(J)$ is a set of labels for the occupied states in the image vector $\Pi_P \nu$.
So finally, with definition (\ref{tau_def}), we can then represent $\mC_P$ and its adjoint $\mC_P^\dagger$ as operators on $\Ct^M$ as
\begin{eqnarray}
\mC_P &=& \sum_{\nu \in \Ft^M} (-1)^{\tau_P (\nu)}| \Pi_P \nu  \rangle \langle \nu|  \label{CP_def}\\
\mC_P^\dagger &=& \sum_{\nu \in \Ft^M}  (-1)^{\tau_P (\nu)} |  \nu \rangle \langle \Pi_P \nu  | 
\end{eqnarray}
These equations are the main result of this section and an overview of the general structure of our mappings is summarized in Fig.\ref{mapping}.
We thus see again a general feature of the mapping from $\mathcal{F}$ to $\Ct^M$; the anti-symmetry properties
are not incorporated in the basis states $| \nu \rangle$ but appear as explicit signs $(-1)^{\tau_P  (\nu)}$ in the transformation operators.
We will derive how these operators transform strings of $X$'s and $Z$'s in the next section.
\begin{figure}
\centering
\includegraphics[width=0.75\textwidth]{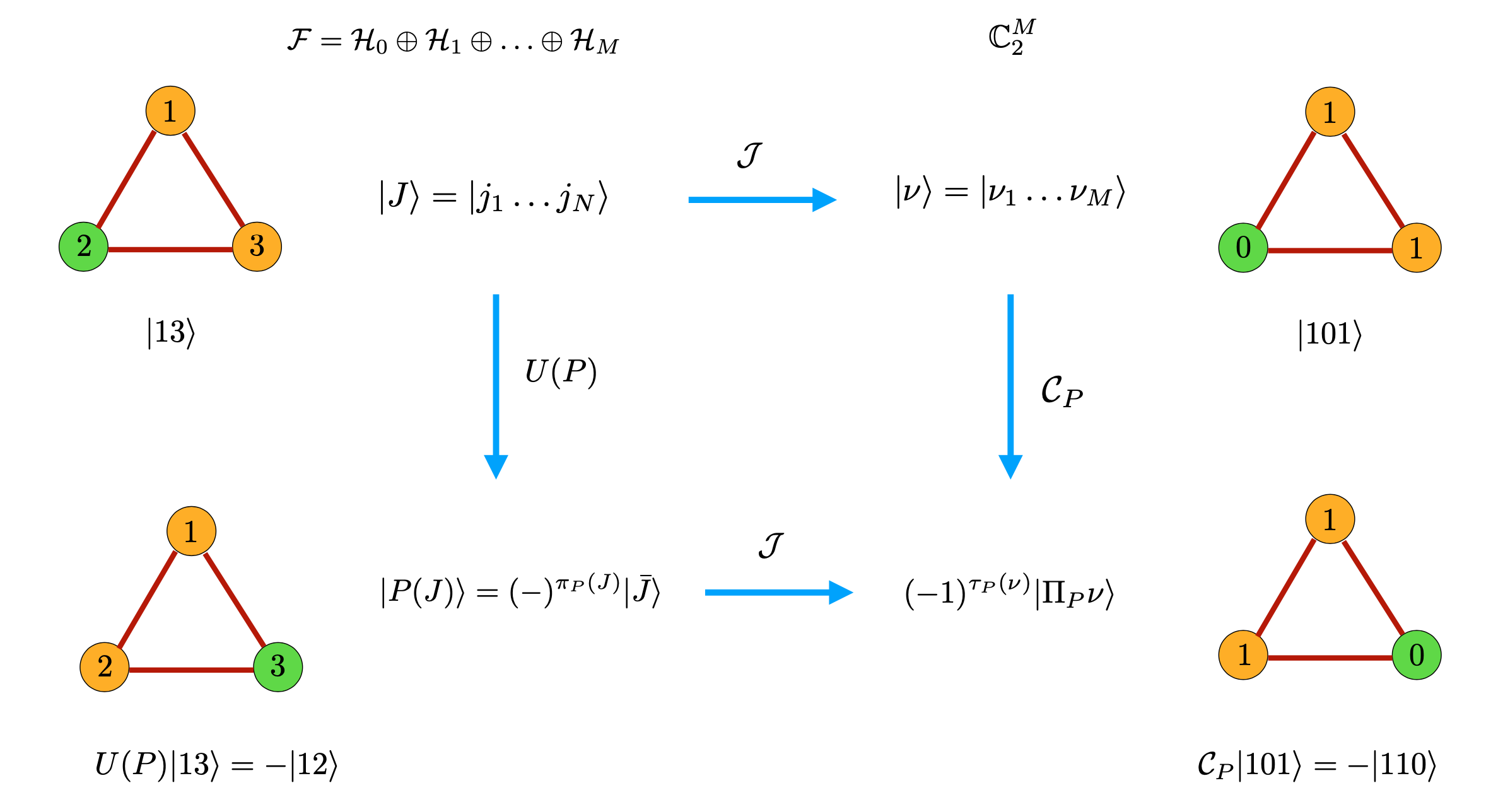} 
\caption{Summary of the mappings from Fock space to multi qubit space. On the left hand side we display the Fock space quantities and on the right hand side  multi-qubit space.
The states are mapped by the transformation $\mathcal{J}$ and the action of $U(P)$ in Fock space is represented by $\mC_P$ in multi-qubit space.}
\label{mapping}
\end{figure}

\subsection{A general theorem}

With the proper definition of $\mC_P$ at hand we are now ready to prove the main result of this work, which is a general theorem on the
transformation of the Pauli tensor products under the permutation symmetries of the second quantized Hamiltonian, which yields a simple practical
method to find their transformation properties.
Take an arbitrary tensor string of Pauli $X$ and $Z$ matrices of the form
\be
\sigma^{q}= (XZ)^q=X^r Z^s = X_1^{r_1} \ldots X_M^{r_M} Z_1^{s_1} \ldots Z_M^{s_M} 
\ee
where $q=(r,s) \in \Ft^{2M}$
. Since
the actions of $X_i$ and $Z_i$ on $| \nu \rangle$ are given by
\be
X_i | \nu \rangle = |\nu_1\ldots  \nu_{i-1}  \nu_i \oplus 1  \nu_i \ldots \nu_M\rangle \quad \quad \quad Z_i | \nu \rangle =(-1)^{\nu_i} | \nu \rangle
\ee
it follows that we can represent the action $\sigma^q$ on a multi qubit state $| \nu \rangle$ in the compact form
\be
\sigma^{q}  | \nu \rangle =X^r Z^s | \nu \rangle = (-1)^{s \cdot \nu} |\nu \oplus r \rangle
\ee
where
\be
s \cdot \nu = \sum_{i=1}^M s_i \nu_i
\ee
denotes the standard Euclidean inner product.
Therefore $\sigma^q$ can be represented in operator form as
\be
\sigma^q = \sum_{\nu \in \Ft^M} (-1)^{s \cdot \nu} |\nu \oplus r \rangle \langle \nu |
\label{sigma_q}
\ee
We can now employ the definition (\ref{CP_def}) of the operator $\mC_P$ to derive the transformation of Eq.(\ref{mC_trans}). Using the
representation (\ref{sigma_q}) and (\ref{CP_def}) we calculate that
\begin{align}
&\mC_P \sigma^q \mC_P^\dagger  = \sum_{x,y,\nu \in \Ft^M}(-1)^{s \cdot \nu + \tau_P (x) + \tau_P (y)}  | \Pi_P x  \rangle \langle x|
\nu \oplus r \rangle \langle \nu |  y \rangle \langle \Pi_P y |  \nonumber \\
 &= \sum_{x,y \in \Ft^M}(-1)^{s \cdot y +  \tau_P (x) + \tau_P (y)}  | \Pi_P x  \rangle \langle x| y \oplus r \rangle  \langle \Pi_P y  |  \nonumber \\
& = \sum_{y \in \Ft^M}(-1)^{s \cdot y +  \tau_P (y \oplus r) + \tau_P (y)}  | \Pi_P y \oplus \Pi_P r \rangle   \langle \Pi_P y  |  \nonumber \\
& = \sum_{\omega \in \Ft^M}(-1)^{s \cdot \Pi_P^{-1} \omega + \tau_P (\Pi_P^{-1} \omega \oplus r)  + \tau_P (\Pi_P^{-1} \omega)} |\omega \oplus  \Pi_P r  \rangle \langle \omega | \nonumber \\
&= X^{\Pi_P r} Z^{(\Pi_P^{-1})^T s}  \sum_{\omega \in \Ft^M} 
(-1)^{ \tau_P (\Pi_P^{-1} \omega \oplus r)  + \tau_P (\Pi_P^{-1} \omega)} | \omega   \rangle \langle \omega |  
= X^{\Pi_P r} Z^{\Pi_P s}  B_{P,r} 
\label{sigma_c_trans_perm}
\end{align}
where we used that $\Pi_P= (\Pi_P^{-1})^T$ since a permutation matrix is orthogonal, and we defined
\be
B_{P,r} = \sum_{\omega \in \Ft^M} 
(-1)^{ \tau_P (\Pi_P^{-1} \omega \oplus r)  + \tau_P (\Pi_P^{-1} \omega)} |\omega  \rangle \langle \omega | 
\label{BPr}
\ee
It remains to investigate this quantity in more detail. It is first of all clear that $B_{P,r}=1$ for $r=0$, and therefore from Eq.(\ref{sigma_c_trans_perm}) it follows that a single $Z_i$ gets mapped to $Z_{P(i)}$
as we have seen in our example, so we only need to figure out the mappings of the $X_i$ operators. 
This is established by the following theorem:
\begin{thrm}
The operator $B_{P,r}$ of Eq.(\ref{BPr}) is given by
\be
B_{P,r} = Z^{Q_P r}
\ee
where $Q_P$ is the $M \times M$-matrix given by
\be
Q_P =  L \Pi_P + \Pi_P L
\label{QP_th}
\ee
where $L$ is the $M \times M$ matrix with every entry filled with ones below the diagonal and zero entries otherwise.
\end{thrm}
\begin{proof}
First of all, since
\be
\mC_P X^{r} \mC_P^\dagger = \prod_{i=1}^M \mC_P X_i^{r_i} \mC_P^\dagger
\label{X_prod}
\ee
where $r_i \in \{ 0,1 \}$, it is sufficient to prove the results for a vector $r$ with a single non-zero entry as the map of a single $X_i$ determines how the whole string gets mapped.
Let us therefore assume $r$ only has a single non-zero entry.
Let $\eta$ denote the sign factor in Eq.(\ref{BPr}), i.e.
\be
\eta = (-1)^{ \tau_P (\Pi_P^{-1} \omega \oplus r)  + \tau_P (\Pi_P^{-1} \omega)} = \det ([ \omega \oplus \Pi_P r ] | [\Pi_P^{-1}\omega \oplus r  ]) \det ( [ \omega ] |[ \Pi_P^{-1} \omega ] )
\label{eq_sigma}
\ee
There are two cases to consider. Let us consider the first case that the set of columns $[\Pi_P^{-1} \omega]$ does not contain the column for which $r$ is non-zero.
Then the first determinant in the product gets an extra row and column, and let us say that this is an $n\times n$-determinant. Let us label the columns $[\Pi_P^{-1} \omega \oplus r  ]$ 
and the rows $[ \omega \oplus \Pi_P r ]$ from ${1, \ldots, n}$. If the non-zero element in $r$ and its image in $\Pi_P r$ appear in column $j$ and row $i$ in this ordering,
then by expanding the determinant in that row and column we have
\be
\det ( [ \omega \oplus \Pi_P r ] | [\Pi_P^{-1}\omega \oplus r  ])  = (-1)^{i+j} \det ( [ \omega ] |[ \Pi_P^{-1} \omega ] )
\ee
and therefore (mod 2) we have
\be
\eta = (-1)^{i+j} = (-1)^{(n-i) + (n-j)}
\ee
 where $n-i$ is the number of rows in $[ \omega \oplus \Pi_P r ]$ above $i$ and $n-j$ the number of occupied columns of the set $[\Pi_P^{-1} \omega \oplus r  ]$ above $j$. 
An example of this procedure is described in Fig.\ref{proof_fig}.
The number $n-j$
can be measured by taking the vector $r$ and filling all entries above the only occupied one with ones and the rest with zeroes, i.e. $r=(0, \ldots,0,\underline{1},0,\ldots,0) \rightarrow (0,\ldots,\underline{0},1,1,\ldots,1)$ and then taking the Euclidean inner product with $\Pi_P^{-1} \omega$, but this amounts 
to calculating
\be
n-j = \Pi_P^{-1} \omega \cdot L r = \omega \cdot \Pi_P L r
\ee
where we used $\Pi_P^{-1}=\Pi_P^T $ since the permutation matrix is orthogonal and 
where $L$ is a $M \times M$ matrix with ones below the diagonal and zeroes otherwise, i.e
\be
L = \begin{pmatrix}
0 & 0 & 0 &\ldots  & 0 \\
1 & 0 & 0 & \ldots  & 0\\
1 & 1 & 0 &\ldots  & 0\\
\vdots & & \ddots & \ddots & \vdots \\
1 & 1 &  \ldots  & 1& 0
\end{pmatrix}
\ee
 Similarly $n-i$ can be found by filling the vector $\Pi_P r$ above its only non-zero entry
with ones and taking all other elements to be zero, and then taking the Euclidean inner product with vector $\omega$; so we have
\be
n- i = \omega \cdot L \Pi_P r
\ee
and consequently
\be
\eta = (-1)^{\omega \cdot (L \Pi_P + \Pi_P L) r}= (-1)^{\omega \cdot Q_P r}
\label{sigma_Q}
\ee
\begin{figure}
\centering
\includegraphics[width=0.75\textwidth]{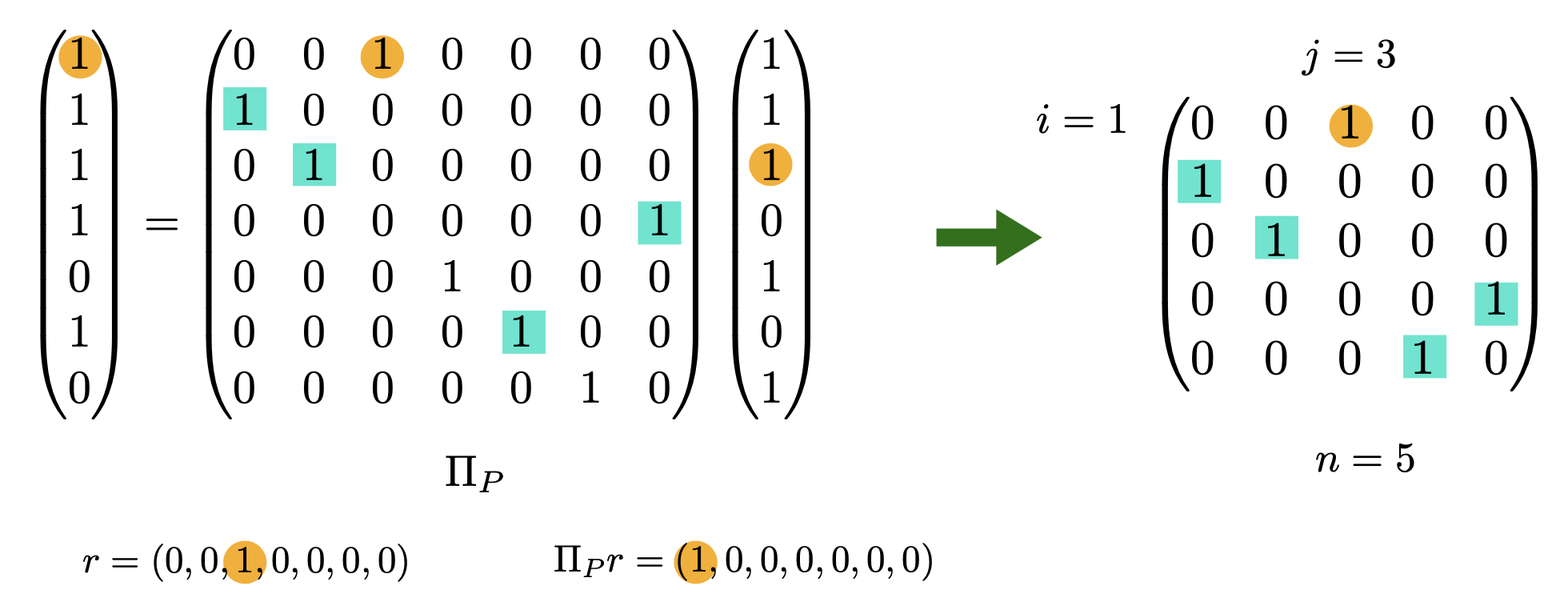} 
\caption{Example for the main step in the proof; for this example we have $(i,j)=(1,3)$, $n=5$ and $2= \Pi_P^{-1} \omega \cdot Lr = 5 -j$ and $4=\omega \cdot L \Pi_P r=5-i$
corresponding to $\eta=(-1)^{4+2}=1$. 
The right hand side of the figure represents the submatrix $( [ \omega \oplus \Pi_P r ] | [\Pi_P^{-1}\omega \oplus r  ])=(12346|12357)$.}
\label{proof_fig}
\end{figure}
where $Q_P = L \Pi_P + \Pi_P L$. The other case is that the set of columns $[\Pi_P^{-1} \omega]$  contains the column for which $r$ is non-zero.
In this case a column (and its corresponding row) is removed in the matrix $( [ \omega \oplus \Pi_P r ] | [\Pi_P^{-1} \omega \oplus r  ])$
as compared to $( [ \omega ] |[ \Pi_P^{-1} \omega  ] )$. Now we can remove the same row and column in $( [ \omega ] |[ \Pi_P^{-1} \omega  ])$
and by the same proof as above we again find the result of Eq.(\ref{sigma_Q}). This finally gives
\be
B_{P,r} = \sum_{\omega \in \Ft^M} 
(-1)^{ \omega \cdot Q_P r} |\omega   \rangle \langle \omega |  = Z^{Q_P r}
\ee
For vectors $r$ with more than one non-zero component we can use Eq.(\ref{X_prod}) repeatedly giving $B_{P,r}=Z^{Q_P r}$ for general $r$,
which establishes what we wanted to prove.
\end{proof}

Using the result of this theorem in Eq.(\ref{sigma_c_trans_perm}) we thus have derived the transformation law for the operator $\sigma^q$ under conjugation
with $\mC_P$. This result constitutes the main theorem of this work:

\begin{thrm}
\label{thrm_symmetry}
Consider a general Hamiltonian in Jordan-Wigner form
\be
\hH = \sum_i \alpha_i X^{r_i }Z^{s_i}
\label{H_alpha}
\ee
where $\alpha_i \in \mathbb{C}$ are coefficients and $r_i, s_i \in \Ft^{M}$. 
Let $P$ be a permutation corresponding to a symmetry of the Hamiltonian and $\Pi_P$ its permutation matrix.
Then the unitary operator
\be
\mC_P |\nu \rangle =(-1)^{\tau_P (\nu)} | \Pi_P \nu \rangle \quad ,\quad (-1)^{\tau_P (\nu)} = \det ([\Pi_P \nu] | [\nu])
\ee
acting on tensor products of qubits $\nu \in \Ft^M$ maps the Hamiltonian to itself according to
\be
\hH = \mC_P \hH \mC_P^\dagger =  \sum_i \alpha_i X^{\Pi_P r_i} Z^{\Pi_P s_i} Z^{Q_P r_i}
\label{H_alpha_symm}
\ee
where $Q_P = L \Pi_P + \Pi_P L$, where $L$ is the $M \times M$ matrix with every entry filled with ones below the diagonal and zero entries otherwise.
The explicit form of the Clifford tableau corresponding to $\mC_P$ is given by
\be
\mathcal{M}_P 
= 
\left(
\begin{array}{c|c}
\Pi_P & 0 \\
\hline
Q_P & \Pi_P
\end{array}
\right)
\label{tableau_thrm}
\ee
\end{thrm}

When we apply this theorem to the example of section \ref{sec_example} we recover precisely the Clifford tableau $\mathcal{M}_P$
of Eq.(\ref{tableau_example}). Interestingly Theorem \ref{thrm_symmetry} can also be applied directly in Jordan-Wigner form to find
Hamiltonians invariant under permutations $P$ and therefore, given a general expansion as in Eq.(\ref{H_alpha}), the coefficients $\alpha_i$ can be chosen
such that Eq.(\ref{H_alpha_symm}) is satisfied.\\
We finally remark that the permutation symmetries induce a finite dimensional representation $P \rightarrow \mathcal{M}_P$ inside the group of symplectic matrices.
Let $P_1$ and $P_2$ be permutations representing symmetries and $\mathcal{M}_{P_1}$ and $\mathcal{M}_{P_2}$ be the corresponding 
Clifford tableaux. Then the action of $P_1$ followed by $P_2$ yields for the Clifford tableaux that
\begin{eqnarray}
\mathcal{M}_{P_2} \mathcal{M}_{P_1}  &=&
\left(
\begin{array}{c|c}
\Pi_{P_2} & 0 \\
\hline
Q_{P_2} & \Pi_{P_2}
\end{array}
\right)
\left(
\begin{array}{c|c}
\Pi_{P_1} & 0 \\
\hline
Q_{P_1} & \Pi_{P_1}
\end{array}
\right) \nonumber \\
&=&
\left(
\begin{array}{c|c}
\Pi_{P_2}  \Pi_{P_1} & 0 \\
\hline
Q_{P_2} \Pi_{P_1} + \Pi_{P_2} Q_{P_1}  & \Pi_{P_2} \Pi_{P_1} 
\end{array}
\right)
\end{eqnarray}
The lower left block has the structure
\begin{eqnarray}
Q_{P_2} \Pi_{P_1} + \Pi_{P_2} Q_{P_1} &=& ( L \Pi_{P_2} + \Pi_{P_2} L) \Pi_{P_1} + \Pi_{P_2}  ( L \Pi_{P_1} + \Pi_{P_1} L) \nonumber \\
&=& L \Pi_{P_2} \Pi_{P_1} + \Pi_{P_2} \Pi_{P_1} L
\end{eqnarray}
since the addition of the matrix $\Pi_{P_2}   L \Pi_{P_1}$ to itself yields zero for $\Ft$ entries.
It thus follows that, since $ \Pi_{P_2} \Pi_{P_1}=  \Pi_{P_2 \circ P_1}$, we have
\be
\mathcal{M}_{P_2} \mathcal{M}_{P_1}  = 
\left(
\begin{array}{c|c}
\Pi_{P_2 \circ P_1} & 0 \\
\hline
L \Pi_{P_2 \circ P_1} + \Pi_{P_2 \circ P_1} L  & \Pi_{P_2 \circ P_1} 
\end{array}
\right)
 = \mathcal{M}_{P_2 \circ P_1}
\ee
Therefore the permutation symmetries of the Hamiltonian induce a representation $P \rightarrow \mathcal{M}_P$ inside the group of symplectic matrices.
This is a finite subgroup of the Clifford group \cite{Mastel2023}, which plays an important role in quantum error correcting formalisms,
so that perhaps the subspaces invariant under permutations can be employed in error correction schemes for quantum computing 
calculations where certain symmetries need to be enforced.

\section{Qubit removal by symmetries}

\subsection{Constructing $Z$-strings commuting with the Hamiltonian}

In this and the following section we will address the issue of removal of qubits based on symmetry properties of the Hamiltonian.
We start by considering the $M \times M$  permutation matrix $\Pi_P$ corresponding to a symmetry  that permutes all the labels of the creation an annihilation operators 
in the Hamiltonian as in Eq.(\ref{H_perm}). Because $\Pi_P$ is a unitary matrix it can be diagonalised by a unitary transformation $V$ as
\be
\lambda_\alpha \delta_{\alpha \beta} = \sum_{ij} V_{\alpha i} \Pi_{P,ij} (V^\dagger)_{j\beta} =  \sum_j V_{\alpha P(j)} (V^\dagger)_{j\beta}
\label{Pi_diag}
\ee
where we used that $\Pi_{P,ij}=\delta_{iP(j)}$. Since every permutation has an finite order, i.e. there is a smallest integer $n$ such that $P^n=1$, it follows
that the eigenvalues $\lambda_\alpha$ satisfy $\lambda_\alpha^n=1$ such that they are the $n$-th roots of unity for some $n$.
We then define new creation and annihilation operators
\be
c_\alpha = \sum_j V_{\alpha j} a_j \quad \quad c_\alpha^\dagger = \sum_j a_j^\dagger (V^\dagger)_{j \alpha}  
\label{b_def}
\ee
that again obey the anti-commutation relations.
We can thus define new anti-symmetric $N$-body ket states
\begin{align}
| \alpha \rangle &= | \alpha_1 \ldots \alpha_N \rangle = c_{\alpha_N}^\dagger \ldots  c_{\alpha_1}^\dagger | 0 \rangle \nonumber \\
&= \sum_{J} (V^\dagger)_{j_N \alpha_N } \ldots  (V^\dagger)_{j_1 \alpha_1} a^\dagger_{j_N} \ldots a^\dagger_{j_1} |0 \rangle
= \sum_{J} (V^\dagger)_{j_N \alpha_N } \ldots  (V^\dagger)_{j_1 \alpha_1} | J \rangle
\end{align}
It then follows that
\begin{align}
U (P) | \alpha \rangle &= \sum_{J} (V^\dagger)_{j_N \alpha_N} \ldots  (V^\dagger)_{ j_1 \alpha_1 } U (P)| J \rangle
= \sum_{J}  (V^\dagger)_{j_N \alpha_N } \ldots  (V^\dagger)_{ j_1 \alpha_1 } | P (J) \rangle \nonumber \\
&=  \sum_{J}  (V^\dagger)_{j_N \alpha_N } \ldots  (V^\dagger)_{ j_1 \alpha_1 }   a_{P(j_N)}^\dagger \ldots a_{P(j_1)}^\dagger |0 \rangle \nonumber \\
&= \sum_{J, \beta}  V_{\beta_N, P(j_N)}(V^\dagger)_{j_N \alpha_N } \ldots  V_{\beta_1 P(j_1)}(V^\dagger)_{ j_1 \alpha_1 }   c_{\beta_N}^\dagger \ldots c_{\beta_1}^\dagger  |0 \rangle \nonumber \\
&= \sum_\beta \delta_{\beta_N \alpha_N} \lambda_{\beta_N} \ldots   \delta_{\beta_1 \alpha_1} \lambda_{\beta_1} | \beta \rangle =  \lambda_{\alpha_1} \ldots \lambda_{\alpha_N} | \alpha \rangle
\label{U_eigen}
\end{align}
where in the last step we used Eq.(\ref{Pi_diag}); so $U (P)$ becomes a diagonal matrix in terms of the new many-body kets  $|\alpha \rangle$.
The matrix elements $\langle \alpha | \hH | \beta \rangle$ of the Hamiltonian in this basis are most conveniently calculated by first transforming the $a_i$ operators to $c_\alpha$-operators
in the second quantized Hamiltonian according to Eq.(\ref{b_def}). \\
Having expressed the Hamiltonian in terms of the $c_\alpha$-operators we can obtain multi-qubit representation of the Hamiltonian
by replacing $c_\alpha$ and $c_\alpha^\dagger$ by equivalent expressions in Jordan-Wigner form, while
Eq.(\ref{U_eigen}) in occupation number representation becomes
\be
\mC_P | \nu_1 \ldots \nu_M \rangle = (\lambda_{\alpha_1})^{\nu_1} \ldots (\lambda_{\alpha_M})^{\nu_M} | \nu_1 \ldots \nu_M \rangle 
\ee
where $\lambda_{\alpha_k}$ is a $n$-root of unity depending under which symmetry class one-body state $| \alpha_k \rangle$ transforms. In the special case 
that $P^2=1$
we have $\lambda_{\alpha_k}=\pm 1$ and we have that
\be
\mC_P  = \sum_\nu (\pm)^{\nu_1} \ldots  (\pm)^{\nu_M} | \nu \rangle \langle \nu | = \sum_\nu (-1)^{\gamma \cdot \nu}  | \nu \rangle \langle \nu |
 = Z_1^{\gamma_1} \ldots Z_M^{\gamma_M}
 \label{mC_Zstring}
\ee
where $\gamma_k=1$ if $\lambda_{\alpha_k}=-1$ and zero otherwise. \\
Let us now consider the largest subgroup of the Hamiltonian that has no group elements of order higher than two; such a subgroup is Abelian and
is called a Boolean group. Since all the group elements commute, we can find a unitary matrix $V$ that diagonalises all permutation matrices $\Pi_P$ at the same time
and use this matrix in Eq.(\ref{b_def}).
%In fact, it is sufficient to do find such a matrix this for the group generators, as it will automatically diagonalise the permutation matrices of the other group elements.
Since $[\hH, \mC_P]=0$ for all such $P$ it follows from Eq.(\ref{mC_Zstring}) that the Hamiltonian commutes with the corresponding
strings of $Z_i$ operators for all $P$ in the subgroup, i.e.
\be
[ \hH , Z_1^{\gamma_1} \ldots Z_M^{\gamma_M} ] = 0
\label{Zstring_H}
\ee
 Such strings are order two, i.e. their square is the unit operator and we will show that by a suitable choice of qubit mapping
there is a transformation that maps them to single $Z_i$-matrices which in the new representation also commute with the Hamiltonian. This means that
in this representation we can discard the corresponding qubits and replace the corresponding $Z_i$-operators in the Hamiltonian by their eigenvalues, i.e. either $1$ or $-1$.
We thereby have succeeded in reducing the qubit requirement of the problem; how the corresponding transformation is constructed in practice is discussed in the next section.\\
We finally remark that apart from the permutation symmetries there are (if there are no spin-coupling terms in the Hamiltonian)
also unitary matrices that square to the identity and which are related to number conservation of each spin species independent of any other spatial symmetry.
The corresponding operators $\mC_\sigma$ are given in multi-qubit space as 
\be
\mC_\sigma  = \sum_\nu (-1)^{n_\sigma \cdot \nu}  | \nu \rangle \langle \nu |
 = Z_1^{n_{\sigma,1}} \ldots Z_M^{n_{\sigma,M}}
 \label{number_op}
\ee
where $n_{\sigma}$ is an $M$-dimensional vector with coefficients $n_{\sigma,j}=1$ for every one-particle state $j$ representing a spin $\sigma$ electron
and with zero coefficients otherwise. 

\subsection{Symmetry-adapted encoding by affine qubit transformations: a short proof}

In order to perform the qubit reduction mentioned in the previous section we need to find a transformation that transforms strings of Pauli $Z$-operators into
single Pauli $Z$-matrices. We now show that an affine transformation on qubits, i.e. an invertible linear transformation $A$ followed by a translation over a vector $b$
induces a unitary transformation that we can choose in such a way to achieve our goal.   
We follow a notation inspired by reference \cite{Hostens2005} which we extend to affine transformations, and define the unitary transformation $\mC$ by
\be
\mC = \sum_{x \in \Ft^M} | Ax \oplus b \rangle \langle x| \quad \quad \quad \mC^\dagger = \sum_{y \in \Ft^M}|  y \rangle \langle Ay \oplus b | 
\ee
and study how a general product $\sigma^q=X^r Z^s$ as in Eq.(\ref{sigma_q}) transforms; we have 
\begin{align}
&\mC \sigma^q \mC^\dagger  = \sum_{x,y,\nu \in \Ft^M}(-1)^{s \cdot \nu}  | Ax \oplus b \rangle \langle x|
\nu \oplus r \rangle \langle \nu |  y \rangle \langle Ay \oplus b |  \nonumber \\
& = \sum_{x,y \in \Ft^M}(-1)^{s \cdot y}  | Ax \oplus b \rangle \langle x| y \oplus r \rangle  \langle Ay \oplus b |  \nonumber \\
& = \sum_{y \in \Ft^M}(-1)^{s \cdot y}  | Ay \oplus Ar \oplus b \rangle   \langle Ay \oplus b |  
 = \sum_{\omega \in \Ft^M}(-1)^{s \cdot A^{-1} (\omega+b)} |\omega \oplus Ar  \rangle \langle \omega | \nonumber \\
& = \sum_{\omega \in \Ft^M}(-1)^{(A^{-1})^Ts \cdot (\omega\oplus b)} |\omega \oplus Ar  \rangle \langle \omega |
= (-1)^{(A^{-1})^Ts \cdot b} X^{Ar} Z^{(A^{-1})^Ts}  \nonumber \\
&= (-1)^{\mu_b (q)} \sigma^{\mathcal{M}_A q}
\label{sigma_c_trans}
\end{align}
where we defined
\be
\mathcal{M}_A = 
\begin{pmatrix}
A & 0 \\
0 & (A^{-1})^T
\end{pmatrix}
\quad \quad \quad \mu_b (q) = (A^{-1})^Ts \cdot b
\label{trans_result}
\ee
This result represents a slightly corrected and considerably shorter proof of a theorem proved in \cite{Picozzi2023} (which erroneously gives $\mu_b (q)=s \cdot b$ instead).\\ 
As a corollary we note that if we choose $s=A^T t$ for some $t \in \Ft^M$, we see from Eq.(\ref{sigma_c_trans}) that conjugation by $\mC$ maps 
$Z^{A^T t}$ to $Z^t$ up to a sign; therefore by choosing vectors $t_i \in \Ft^M$ to have entry $1$ on position $i$ and zero otherwise we can map to a single $Z_i$.
The vectors $u_i=A^T t_i$ are just the columns of $A^T$ or, equivalently, the rows of $A$. This means that for any vector $u_i=(A_{i1}, \ldots, A_{iM})$ corresponding to a row vector $i$ of $A$
the string $Z^{u_i}$ is mapped to a single Pauli matrix $Z_i$ with a phase, according Eq.(\ref{trans_result}), corresponding to $(-1)^{t_i \cdot b}$, and thus
\be
\mC Z^{u_i}  \mC^\dagger = (-1)^{t_i \cdot b} Z_i    
\label{Zi_map}
\ee
We can summarize our results in the following theorem
\begin{thrm}
Let $\sigma^q =X^r Z^s$ with $q=(r,s)$ and $r,s \in \Ft^M$. Further let $b \in \Ft^M$ and $A$ be in invertible $M \times M$-matrix with entries in $\Ft$. Then the unitary transformation
\be
\mC = \sum_{\nu \in \Ft^M} | A\nu \oplus b \rangle \langle \nu|
\ee   
on $\Ct^M$ has the property $\mC c^q \mC^\dagger = (-1)^{\mu_b (q)} \sigma^{\mathcal{M}_A q}$ where
\be
\mathcal{M}_A = 
\begin{pmatrix}
A & 0 \\
0 & (A^{-1})^T
\end{pmatrix}
\quad \quad \quad \mu_b (q) = (A^{-1})^Ts \cdot b
\label{MA_mat}
\ee
In particular if $Z^{u_i}$ is a string of Pauli $Z$-operators formed by row $i$ of matrix $A$, i.e. $u_i=(A_{i1}, \ldots, A_{iM})$ 
and $t_i$ a vector with coefficient $1$ on position $i$ and zero coefficients otherwise, then
\be
\mC Z^{u_i}  \mC^\dagger = (-1)^{t_i \cdot b} Z_i   
\ee
such that the
string $Z^{u_i}$  is mapped to a single Pauli $Z_i$.
\end{thrm}

Now that we know how to map tensor strings of $Z$-matrices to a single $Z_i$ it remains to find a procedure to find such $Z$-strings that commute with the Hamiltonian
as in Eq.(\ref{Zstring_H}). In this we follow a suggestion given in reference \cite{Bravyi2017} (see also \cite{Gunderman2024} for proof of optimality) that we briefly repeat here for convenience.
In Jordan-Wigner form the Hamiltonian attains the form (as we can always express $Y$ in terms of $XZ$)
\be
\hH = \sum_i \alpha_i (XZ)^{q_i}
\ee
where $\alpha_i \in \mathbb{C}$ are coefficients and $q_i =(r_i, s_i) \in \Ft^{2M}$. Then the commutator of $(XZ)^{q^\prime}$ with $q^\prime =(r^\prime, s^\prime)$
with $\hH$ is obtained from the anti-commutation relations $XZ=-ZX$ as follows
\be
[ \hH, (XZ)^{q^\prime} ] = \sum_i \alpha_i [ (XZ)^{q_i}, (XZ)^{q^\prime} ] = \sum_i \alpha_i ((-1)^{ s_i \cdot r^\prime} -(-1)^{r_i \cdot s^\prime}    ) (XZ)^{q_i + q_i^\prime}
\ee
The commutator therefore vanishes if we can find $r^\prime$ and $s^\prime$ such that 
\be
0 = r_i \cdot s^\prime + s_i \cdot r^\prime
\label{symp}
\ee
for all terms $q_i$. This means that the vector $(s^\prime, r^\prime)$ is in the kernel or nullspace of a matrix $K$ with has all the vectors $(r_i, s_i)$ as rows.
The construction of this kernel is precisely the procedure outlined in \cite{Bravyi2017}, who then proceed to design a transformation that maps Pauli strings deduced from the kernel vectors
to single Pauli
$X$-matrices. We will now deviate from that reference as their additional computational 
problem was very much simplified by the alternative method of Picozzi and Tennyson \cite{Picozzi2023}
for which we gave the much shortened derivation above. Instead of mapping to single $X$-matrices the map is to single $Z$-matrices while a numerical construction as in \cite{Bravyi2017} is avoided.
To elucidate the general theory we give an example of the procedure in the next section, in which we
furthermore also illustrate the theory of the earlier sections.

\subsection{Example: removing three qubits for the Hubbard dimer}

We can now illustrate all the theory from the preceding two sections with a simple interacting quantum system, namely the Hubbard dimer.
We label the creation and annihilation operators with the labels $(1,2,3,4)=(\uparrow  1, \uparrow 2 ,\downarrow 1,\downarrow 2)$
where  each entry $\sigma \, i$ inside the last brackets stands for spin $\sigma \in \{ \uparrow , \downarrow \}$ on Hubbard site $i \in \{ 1,2 \}$.
Then the Hubbard Hamiltonian takes the form
\be
\hat{H} = - t (\ad_1 a_2+  \ad_2 a_1 + \ad_3 a_4 + \ad_4 a_3) + U (\ad_1 a_1 \ad_3 a_3 +  \ad_2 a_2 \ad_4 a_4 )
\label{Hubbard_dimer}
\ee
and which transformed Jordan-Wigner notation has the form
\begin{align}
\hH =& - \frac{t}{2} ( X_1 X_2 (1- Z_1 Z_2)  + X_3 X_4 ( 1 - Z_3 Z_4))  \nonumber \\
&+ \frac{U}{4} ( 2 - Z_1-Z_2- Z_3 - Z_4 + Z_1 Z_3 + Z_2 Z_4 )
\label{Hubbard1}
\end{align}
The conservation of number of particles with a given spin quantum number makes the Hamiltonian commute with the operators
\be
\mC_\uparrow = Z_1 Z_2 \quad \quad \quad \mC_\downarrow = Z_3 Z_4
\ee
corresponding to $n_\uparrow =(1,1,0,0)$ and $n_\downarrow =(0,0,1,1)$ in Eq.(\ref{number_op}).
However, it is readily seen that there are additional order two operators related to permutation of the Hubbard sites
and in particular we consider the Boolean subgroup given by the unit permutation and $P(1234) = (2143)$. The latter permutation is represented by a
permutation matrix $\Pi_P$ and a $Q_P$ matrix (calculated from Eq.(\ref{QP_th})) with the explicit form
\be
\Pi_P 
= 
\begin{pmatrix}
0 & 1 & 0 & 0 \\
1 & 0 & 0 & 0 \\
0 & 0 & 0 & 1 \\
0 & 0 & 1 & 0 
\end{pmatrix}
\quad \quad  \quad 
Q_P 
= 
\begin{pmatrix}
1 & 0 & 0 & 0 \\
0 & 1 & 0 & 0 \\
0 & 0 & 1& 0 \\
0 & 0 & 0& 1 
\end{pmatrix}
\ee
From the corresponding Clifford tableau $\mathcal{M}_P$ of Eq.({\ref{tableau_thrm}) it then follows that under the symmetry operation the $X_i$ transform as
\begin{align}
& \mC_P X_1 \mC_P^\dagger = X_2 Z_1\quad \quad \quad  \mC_P X_2 \mC_P^\dagger = X_1 Z_2  \nonumber \\
& \mC_P X_3 \mC_P^\dagger = X_4 Z_3 \quad \quad \quad \mC_P X_4 \mC_P^\dagger = X_3 Z_4 \nonumber 
\end{align}
while $\mC_P Z_i \mC_P^\dagger =Z_{P(i)}$, which indeed is readily seen to preserve the symmetry of the Hamiltonian.
If we diagonalize $\Pi_P$ we find using Eq.(\ref{b_def}) that the new symmetry adapted creation and annihilation operators are given by
\be
c_{1,2} =  \frac{1}{\sqrt{2}} (a_1 \pm a_2) \quad \quad
c_{3,4} =  \frac{1}{\sqrt{2}} (a_3 \pm a_4) \nonumber 
\ee
where $c_1, c_3$ are the + combinations and $c_2,c_4$ are the - combinations of the $a_i$ on the right hand side of the equations. The corresponding equations for the $c_i^\dagger$ are obtained by taking the adjoint of these expressions.
By exposing the symmetry character of the one-particle states we have the following correspondence between symmetry labels and the operators $c_i$:
\be
(c_1,c_2,c_3,c_4) \leftrightarrow ( g_{\uparrow}, u_{\uparrow}, g_{\downarrow}, u_{\downarrow})
\label{symmetries}
\ee
where $g$ and $u$ are gerade (even) and ungerade (odd) symmetries, i.e. under the symmetry $P$ the $g \rightarrow g$ and
$u \rightarrow -u$. So, for example,
the application of $c_1$ removes a spin up particle from a gerade orbital, etc.
The second quantized Hamiltonian in terms of the new creation and annihilation operators attains the form
\begin{align}
\hH =& - t (c_{1}^\dagger c_1 - c_2^\dagger c_2 + c_{3}^\dagger c_3 - c_4^\dagger c_4) \nonumber \\
&+ \frac{U}{2} (c_{1}^\dagger c_1 + c_2^\dagger c_2) (c_{3}^\dagger c_3 + c_4^\dagger c_4) +  \frac{U}{2}  (c_{1}^\dagger c_2 + c_2^\dagger c_1) (c_{3}^\dagger c_4 + c_4^\dagger c_3)
\end{align}
and further translated to Jordan-Wigner form this becomes
\begin{align}
\hH =& - \frac{t}{2} (-Z_1 + Z_2 -Z_3 + Z_4) + \frac{U}{8} (2- Z_1 -Z_2)(2-Z_3-Z_4) \nonumber \\
&+ \frac{U}{8} X_1 X_2 (1 - Z_1 Z_2)   X_3 X_4 ( 1- Z_3 Z_4)
\end{align}
To calculate the operators that commute with this Hamiltonian we can construct the matrix $K$ containing the row vectors $(r_i,s_i)$ for each term
$X^{r_i} Z^{s_i}$ in the Hamiltonian and find its kernel (see the discussion following Eq.(\ref{symp})). In our case we find
%For our Hamiltonian
%the matrix $K$ attains the form
%\be
%K =
%\left(
%\begin{array} {cccc | cccc }
%1 & 1 & 1 & 1 &  0& 0 & 0 & 0  \\
%1 & 1 & 1 & 1 &  0 & 0 & 1 & 1  \\
%1 & 1 & 1 & 1 & 1 & 1  & 0 &0   \\
%1 & 1 & 1 & 1 & 1 & 1 & 1 & 1 \\
%0 & 0 & 0 & 0& 1 & 0 & 0 & 0 \\
%0 & 0 & 0 & 0& 0 & 1 & 0 & 0 \\
%0 & 0 & 0 & 0& 0 & 0 & 1 & 0 \\
%0 & 0 & 0 & 0& 0 & 0 & 0 & 1 \\
%0 & 0 & 0 & 0& 0 & 1 & 0 & 1 \\
%0 & 0 & 0 & 0& 1 & 0 & 1 & 0 \\
%0 & 0 & 0 & 0& 0 & 1 & 1 & 0 \\
%0 & 0 & 0 & 0& 1 & 0 & 0 & 1 
%\end{array} 
%\right)
%\ee
%for which we have 
\be
\ker K = \textrm{span} ( (1,1,0,0,0,0,0,0), (0,1,0,1,0,0,0,0),(0,0,1,1,0,0,0,0) )
\ee
corresponding to the $Z$-strings $Z_1 Z_2$, $Z_2 Z_4$ and $Z_3 Z_4$ that commute with the Hamiltonian.
The first three rows of $A$ for our symmetry encoding must therefore contain the $Z$-part of these kernel vectors, i.e. we have the rows $(1,1,0,0),(0,1,0,1),(0,0,1,1)$ 
for the first three rows of $A$. We take the last row to be linearly independent of these, for which a simple choice is $(0,0,0,1)$. 
In reference \cite{Picozzi2023} the vector $b$ of the affine encoding is chosen in such a way that the eigenvalues of the $Z_i$ in the new representation can be taken to be equal to 1 for the symmetry
sector they are interested in.   
We simply take the translation vector $b=0$ for the affine mapping as we do not want to select particular eigenvalues yet for the group generators transformed to $Z_i$.
With this
we find for the Clifford tableau of Eq.(\ref{MA_mat}) that
\be
\mathcal{M}_A = 
\begin{pmatrix}
A & 0 \\
0 & (A^{-1})^T
\end{pmatrix}
=
\left(
\begin{array} {cccc | cccc }
1 & 1 & 0 & 0 &  &  &  &   \\
0 & 1 & 0 &1 &  &  &  &   \\
0 & 0& 1 & 1 &  &  &  &   \\
0 & 0 & 0 & 1 &  &  &  &  \\
\hline\
 &  &  &  & 1 & 0 & 0 & 0 \\
 &  &  &  & 1 & 1 & 0 & 0  \\
 &  &  &  & 0 & 0 & 1 & 0  \\
 &  &  &  & 1 & 1 & 1 & 1  
\end{array}  \right)
\label{Mmatrix_2}
\ee
From this we can read off the transformations of the $X$ and $Z$ operators, and we indeed find that $Z_1 Z_2, Z_2 Z_4, Z_3 Z_4 \rightarrow Z_1,Z_2,Z_3$
so that the new Hamiltonian $\hH^\prime =\mC \hH \mC^\dagger$ becomes
\begin{align}
\hH^\prime =& - \frac{t}{2} ( -Z_1 Z_2  Z_4 +  Z_2  Z_4 -  Z_3 Z_4 + Z_4 ) + \frac{U}{8} (2- Z_1 Z_2  Z_4 -Z_2 Z_4)(2-Z_3 Z_4-Z_4) \nonumber \\
&+  \frac{U}{8} X_4( 1 -  Z_1 -  Z_3 +  Z_1 Z_3)
\end{align}
which now contains a single $X_4$-operator and otherwise only $Z_i$. We
readily see that the Hamiltonian indeed commutes with $Z_1, Z_2$ and $Z_3$ so that all these operators can be replaced by their eigenvalues $\pm 1$ such that the Hamiltonian
reduces to a simple single qubit Hamiltonian. 
The $2$-particle $S_z=0$ multi-qubit states of gerade symmetry $|x\rangle =|1010 \rangle$ and $|x\rangle =|0101 \rangle$ (see also Eq.(\ref{symmetries})) are mapped by $A$ to 
$| Ax \rangle = |1010\rangle $ and $| Ax \rangle =|1011\rangle$ and therefore from the action of
the $Z_1,Z_2$ and $Z_3$ matrices on the first three entries of these vectors we see that we have $Z_1=Z_3=-1$ and $Z_2=1$ as eigenvalues. The  $2$-particle $S_z=0$ multi-qubit states of ungerade symmetry $|x\rangle=|0110\rangle$ and $|x\rangle = |1001\rangle$ are mapped
to $| Ax \rangle =|1110\rangle$ and $|Ax\rangle=|1111\rangle$ corresponding to $Z_1=Z_3=Z_2=-1$ as eigenvalues. So in the $S_z=0$ sector, where $Z_1=Z_3=-1$ we have the form
\begin{align}
\hH^\prime = - t (1+Z_2) Z_4 + \frac{U}{2} (1 + X_4) 
\end{align}
For the gerade states we have $Z_2=1$ and then:
\be
\hH^\prime
= 
\begin{pmatrix}
- 2t + \frac{U}{2} & \frac{U}{2} \\
\frac{U}{2} & 2t +\frac{U}{2} 
\end{pmatrix}
\ee
We have for the eigenvalues $\lambda = \frac{1}{2} ( U \pm \sqrt{16t^2 + U^2})$ containing the well-known singlet gerade ground state energy of the Hubbard dimer
which is a combination of a Slater determinant with two gerade (bonding) orbitals and a Slater determinant with two ungerade (anti-bonding) orbitals. 
On the other hand for the ungerade states we have $Z_2=-1$ and 
\be
\hH^\prime = 
\frac{U}{2}
\begin{pmatrix}
 1 & 1 \\
 1 & 1 
\end{pmatrix}
\ee 
which has eigenvectors $(1,1)$ with eigenvalue $U$ and $(1,-1)$ with eigenvalue $0$. A similar procedure as discussed above applies to other choices of the
eigenvalues of $Z_1,Z_2$ and $Z_3$; there are six other choices that each lead to a single qubit problem and therefore our original $16 \times 16$-problem has blocked into
eight $2 \times 2$ single qubit problems by using symmetry. In this example we found a large reduction from four qubits to a single one.
In general applications one studies larger systems in which the relative reduction is less spectacular, but the universal procedure for qubit reduction as outlined here nevertheless
applies the same way.

\section{Conclusions and outlook}

In this work we considered the symmetry properties of many-body Hamiltonians and the way they are encoded in a Jordan-Wigner transformed manner for implementations on quantum computers.
Whereas these symmetries are explicit in the original second quantized Hamiltonian defined on Fock space, they are concealed in its Pauli tensor form when the Hamiltonian is translated
to multi-qubit space. We proved a general theorem that allows us to straightforwardly derive the symmetry transformation properties of the Pauli operators and thereby making explicit again the symmetries 
of the Hamiltonian. We further gave a short and simplified proof of a recently introduced qubit reduction scheme based on affine mappings of qubits, and illustrated all our theory with an example.\\
The results presented here may be extended via several avenues of research. Instead of regarding the symmetry properties of Jordan-Wigner transformed encodings we could study various other
forms of encoding such as the Bravyi-Kitaev transformation \cite{Seeley2012} or perhaps even more general ones \cite{Steudtner2018}, and ask for a generalization of our main theorem to that case; at first glance there does not seem 
to be any major obstruction in our proofs that would prevent such a generalization of the main theorem. Such a transformation may, for example, be optimized to reduce circuit depth for noise reduction on noisy quantum devices.\\
Regarding the affine encoding scheme
we could wonder if our results could be generalized to non-Boolean but still Abelian groups with group elements of higher order than two. Some of our equations do remain valid and an extension
to general Abelian groups would widen the range of applications for qubit reduction. At least if we are willing to replace qubits by qudits \cite{Hostens2005} in $\mathbb{F}_q$, where $q$  is the largest order of a group element in the Abelian group that we consider, a symmetry reduction of qudits seems an achievable goal. \\
Finally there may be several connections to error correction schemes in quantum computing. The symmetry operators
$\mC_P$ that we discussed were elements of the Clifford group used to develop stabilizer codes \cite{Aaronson2004,VanDenBerg2021,Mastel2023} and in our case we find subspaces of $\Ct^M$ stabilized by symmetry operations, so that perhaps a useful connection
can be made to ways to correct symmetry errors in quantum computing schemes.   
\\
\ack
I wish to acknowledge support from the Finnish Academy under project number 356906. I further like to thank Ivano Tavernelli 
for insightful discussions and hosting me at IBM Research Europe in R\"uschlikon, Switzerland where part of this work was presented.\\

\bibliography{Quantumcomputing}
\bibliographystyle{unsrt}

\end{document}